\documentclass[aps,pra,letterpaper,reprint,twocolumn,floatfix,superscriptaddress,notitlepage,usenames,dvipsnames,svgnames,x11names,table,nofootinbib,longbibliography]{revtex4-2}

\pdfoutput=1
\usepackage{graphicx, color, graphpap}      % Include figure files
\usepackage{amsmath}
\usepackage{enumerate}
\usepackage{amssymb}
\usepackage{amsthm}
\usepackage{pstricks}
\usepackage{float}

\usepackage[pdfusetitle]{hyperref}
\hypersetup{pdflang={English},colorlinks=true,linkcolor=RoyalBlue,citecolor=RoyalBlue,urlcolor=ForestGreen,breaklinks=true}

\usepackage[T1]{fontenc}
\usepackage{bbm}
\usepackage{dsfont}
\usepackage[linesnumbered,ruled,vlined]{algorithm2e}
\SetKwInput{kwInit}{Init}
\usepackage{mathtools}

\usepackage{tikz}
\usetikzlibrary{positioning}
\usepackage{xcolor}
\usepackage[caption=false]{subfig}
\usepackage[ruled,vlined]{algorithm2e}
\usepackage{siunitx}
\usepackage{qcircuit}
\usepackage{pifont} % http://ctan.org/pkg/pifont
\usepackage{physics}
\usepackage{thmtools}
\usepackage{thm-restate}


\renewcommand{\eqref}[1]{(\ref{#1})}
\newtheoremstyle{example}{\topsep}{\topsep}%
{}%         Body font
{}%         Indent amount (empty = no indent, \parindent = para indent)
{\bfseries}% Thm head font
{:}%        Punctuation after thm head
{   }%     Space after thm head (\newline = linebreak)
{\thmname{#1}\thmnumber{ #2}}%\thmnote{ #3}}%         Thm head spec
\theoremstyle{example}
\newtheorem{theorem}{Theorem}

\theoremstyle{definition}

\newtheorem*{theorem*}{Theorem}

\def\orcid#1{\kern -0.4em\href{https://orcid.org/#1}{\includegraphics[keepaspectratio,width=0.7em]{orcid_logo.pdf}}}

\renewcommand{\H}{\mathcal{H}}

\usepackage{lipsum}

\long\def\ca#1\cb{} %Use for commenting out: \ca...\cb
\begin{document}
\title{Digital quantum state preparation by nucleation}

\author{Jean Paul Sadia}
\affiliation{Department of Computational Mathematics, Science, and Engineering, Michigan State University, East Lansing, MI 48824, USA}
\affiliation{Center for Quantum Computing, Science, and Engineering, Michigan State University, East Lansing, MI 48824, USA}

\author{Morten Hjorth-Jensen}
\affiliation{Department of Physics and Center for Computing in Science Education, University of Oslo, N-0316 Oslo, Norway}

\author{Dean Lee}
\affiliation{Facility for Rare Isotope Beams and Department of Physics and Astronomy,
Michigan State University, East Lansing, MI 48824, USA}
\affiliation{Department of Physics and Astronomy, Michigan State University, East Lansing, MI 48824, USA}

\author{Ryan LaRose}
\thanks{Corresponding author:  \href{rmlarose@msu.edu}{rmlarose@msu.edu}}
\affiliation{Department of Computational Mathematics, Science, and Engineering, Michigan State University, East Lansing, MI 48824, USA}
\affiliation{Department of Electrical and Computer Engineering, Michigan State University, East Lansing, MI 48824, USA}
\affiliation{Department of Physics and Astronomy, Michigan State University, East Lansing, MI 48824, USA}
\affiliation{Center for Quantum Computing, Science, and Engineering, Michigan State University, East Lansing, MI 48824, USA}

\begin{abstract}
    Quantum state preparation is crucial for quantum algorithms, and improved state preparation methods can reduce resource requirements by orders of magnitude. In this work, we introduce a state preparation method for lattice Hamiltonians on digital quantum computers inspired by nucleation. In our method, the exact ground state of a small lattice Hamiltonian is prepared in a quantum circuit, then a series of growth stages are performed via Trotterized adiabatic evolution to increase the size of the lattice to the desired model. Using the adiabatic theorem and Trotter error bounds, we provide sufficient conditions on the total evolution time and number of Trotter steps required for nucleation to prepare an initial state with desired fidelity. In addition, we introduce optimized nucleation which considers variational optimization on top of the nucleation structure. For a two-dimensional Ising model, we show that optimized nucleation can outperform a current leading method for state preparation on digital quantum computers.
\end{abstract}

% =============================================================================
% =============================================================================
\maketitle
% =============================================================================
% =============================================================================

\section{Introduction}

\begin{figure*}
    \centering
    \includegraphics[width=\linewidth]{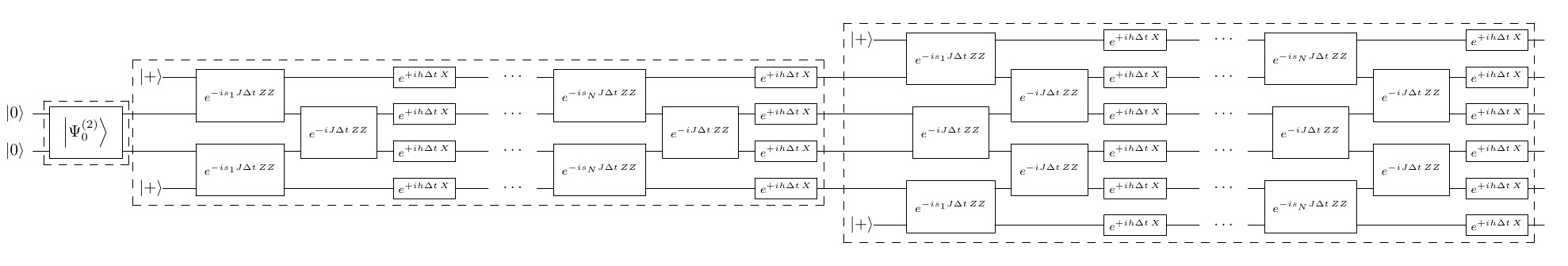}
    \caption{The quantum circuit for nucleation demonstrated for the one-dimensional transverse-field Ising model (TFIM)~\eqref{eqn:tfim}. First, the exact ground state of a small lattice Hamiltonian, here for two qubits, is found and prepared on a quantum computer. Then, we grow the model from a small to a large lattice, here adding two qubits to the edges of the lattice initialized to the $|+\rangle$ state at each step. Each growth stage adiabatically interpolates from the smaller previous lattice to the larger subsequent lattice with new qubits added~\eqref{eqn:nucleation-growth-step}. To do this in a quantum circuit, we use a discretized schedule~\eqref{eqn:schedule} and a number of Trotter steps $N$. Each growth stage evolves for a total evolution time $T$, with $\Delta t = T / N$.}
    \label{fig:nucleation-circuit}
\end{figure*}

Many quantum algorithms require an accurate initial state to guarantee convergence. Such algorithms include phase estimation~\cite{kitaev_quantum_1995}, quantum imaginary time evolution~\cite{motta_determining_2020} quantum Krylov methods~\cite{stair_multireference_2020,epperly_theory_2022,shen_real-time_2023,kirby_analysis_2024}, and the rodeo algorithm~\cite{choi_rodeo_2021}, to name a few. Physically, the problem is known as the (Van-Vleck) orthogonality catastrophe~\cite{van_vleck_nonorthogonality_1936,anderson_infrared_1967,mcclean_exploiting_2014}, and it has long been identified as a major consideration for fault-tolerant quantum algorithms at scale~\cite{mcclean_exploiting_2014}, with recent work~\cite{fomichev_initial_2024} demonstrating orders of magnitude reduction in computational resources via improved state preparation. Because of this, many methods have been proposed for quantum state preparation including variational quantum algorithms~\cite{khatri_quantum-assisted_2019,cerezo_variational_2021}, tensor network methods~\cite{schoen_sequential_2007,ran_encoding_2020,kukliansky_qfactor_2023}, and (by some of the present authors) resolution refinement for lattice models~\cite{bogner_quantum_2026}. While each method has its own advantages and disadvantages depending on the application, the large suite of methods reflects the importance of state preparation as a general routine for quantum computers.

Here, we introduce a method for state preparation on digital quantum computers inspired by nucleation. The idea of nucleation has a long history in physics, with the concept first appearing in thermodynamics by Gibbs in 1878~\cite{gibbs_equilibrium_1874}. In 1926, Volmer and Weber~\cite{volmer_keimbildung_1926} simulated nucleation in supersaturated systems and developed the classical theory of nucleation. After, this was formulated into a statistical mechanical theory by Langer in the 1960s and 1970s~\cite{langer_theory_1967,langer_statistical_1969,langer_metastable_1974}. Wilson's numerical renormalization group (NRG) from 1975~\cite{wilson_renormalization_1975} and subsequently White's density matrix renormalization group (DMRG) from 1992~\cite{white_density_1992} encapsulate the nucleation and growth steps through diagonalization of a small seed and iterative addition of new sites. In the 2000s, this broad framework was translated into the modern language of matrix product states~\cite{vidal_efficient_2003,schollwoeck_density-matrix_2011}, from which the generalization to other tensor networks such as the multiscale entanglement renormalization ansatz (MERA)~\cite{vidal_class_2008} arose. Subsequently, many methods within the family have been developed. Notably, in 2016 Swingle and McGreevy introduced the $s$-source framework in which a ground state of a lattice with $L$ sites grows to the ground state of a lattice with $2L$ sites from $s$ copies of the smaller lattice by a quasi-local unitary transformation~\cite{swingle_renormalization_2016}. In 2020, Olund \textit{et al.} implemented and characterized an algorithm inspired by this $s$-source framework for building quantum many-body ground states, in particular by treating the $s$-source framework as a variational ansatz and then optimizing local unitary transformations via tensor network methods~\cite{olund_adiabatic_2020}.

In this work, we introduce a method for state preparation of lattice Hamiltonians on digital quantum computers which we call nucleation. Nucleation starts with a small lattice (the \textit{seed} or \textit{nucleus}) that can be solved exactly, then iteratively grows to a larger lattice. Nucleation growth stages are done by Trotterized adiabatic evolution written explicitly as a digital quantum circuit. We provide sufficient conditions on the total evolution time and number of (first-order) Trotter steps required to prepare an initial state with desired fidelity. From this, exact gate counts can be determined. While these resource conditions are sufficient, we find numerically that nucleation can prepare high fidelity initial states with significantly fewer resources than expected from these conditions. Motivated by current and near-term quantum hardware constraints, we also introduce \textit{optimized nucleation} where we apply the nucleation circuit structure and initial parameters, then variationally optimize parameters to minimize energy. This provides a tunable tradeoff between Trotterized nucleation, which is efficient  but may have high overhead, and variational optimization, which has low overhead but requires a hard optimization problem. Notably, we show that optimized nucleation can outperform a leading state-of-the-art method, converting a matrix product state (MPS) to a quantum circuit~\cite{ran_encoding_2020,mpstocircuit2025}, for preparing initial states with a two-dimensional Ising model. Indeed, optimized nucleation produces more accurate ground state approximations with fewer two-qubit gates, even when we allow for variational optimization in the MPS to circuit method. For these reasons, we  expect that nucleation will prove useful for initial state preparation on both near-term and fault-tolerant digital quantum computers. 

While our work builds on the long history of nucleation previously discussed and shares common elements with~\cite{swingle_renormalization_2016,olund_adiabatic_2020}, it is distinct in several ways. First, we focus on state preparation for digital quantum computers and provide a complete, non-variational quantum circuit for state preparation which has not appeared in prior work. We further provide rigorous worst-case bounds on gate counts required to achieve a desired target fidelity, which also are not present in prior work to the best of our knowledge. While~\cite{swingle_renormalization_2016,olund_adiabatic_2020} consider doubling the lattice size at each stage, our nucleation procedure is closer to Wilson's NRG~\cite{wilson_renormalization_1975}, White's DMRG~\cite{white_density_1992}, and the site-by-site procedure of~\cite{moosavian_site-by-site_2019} in that we introduce ancilla qubits at neighboring sites in the growth stage. Second, motivated by current and near-term quantum hardware constraints, we introduce optimized nucleation and show that it can significantly reduce the number of required gates. Optimized nucleation is related to the variational procedure of~\cite{olund_adiabatic_2020}, but starts from our (non-variational) nucleation structure, provides an initial ``warm start'' parameterization through Trotterization of a specific adiabatic schedule, and introduces a layer-wise optimization procedure~\cite{skolik_layerwise_2021,larose_mixer-phaser_2022}. In the context of variational quantum state preparation, optimized nucleation provides both a structure and initial parameterization, which are often nontrivial tasks for other ansatzes that can introduce scaling issues~\cite{mcclean_barren_2018}. Finally, we provide a direct benchmark to another leading algorithm for state preparation, converting a matrix product state to a circuit, and show that optimized nucleation can outperform this method for a two-dimensional Ising model. We remark that nucleation is distinct from hierarchical fusion methods for state preparation~\cite{patkowski_hierarchical_2026}, which merge solutions to smaller lattice problems together rather than grow an initial seed, and distinct from resolution refinement for state preparation~\cite{bogner_quantum_2026}, in which the domain size is fixed (no growth occurs) but the resolution of the lattice iteratively increases.

In what follows, we introduce nucleation in Sec.~\ref{sec:nucleation} and prove sufficient conditions for its resource requirements. We then introduce optimized nucleation in Sec.~\ref{sec:optimized-nucleation}, and finally present results of numerical simulations and benchmarks with other state preparation algorithms in~\ref{sec:numerical-results}.

\section{Nucleation} \label{sec:nucleation}

While nucleation applies to any lattice Hamiltonian (satisfying the conditions of Theorem~\ref{thm:adiabatic-time}), for clarity of exposition we consider the transverse-field Ising model (TFIM)
\begin{equation} \label{eqn:tfim}
    \H =  J \sum_{\langle ij \rangle} Z_i Z_j - h \sum_i X_i .
\end{equation}
Here, $J$ is the coupling strength, $h$ is the transverse field strength, $X$ and $Z$ are the Pauli matrices
\begin{equation}
    X = \left[ \begin{matrix}
        0 & 1 \\
        1 & 0
    \end{matrix} \right] 
    \qquad
    Z = \left[ \begin{matrix}
        1 & 0 \\
        0 & -1
    \end{matrix} \right] ,
\end{equation}
and $\langle i j \rangle$ extends over all connected pairs. The quantum circuit for nucleation is shown in Fig.~\ref{fig:nucleation-circuit} for the one-dimensional TFIM. At the start of nucleation, we prepare the exact ground state of a small lattice, shown for two sites $|\Psi_0^{(2)} \rangle$. This state is found exactly on a classical computer and then compiled into a quantum circuit via standard methods. Then, the first growth or nucleation stage occurs. Here, two additional qubits initialized to the $|+\rangle$ state are introduced on both ends of the lattice. Then, we adiabatically grow from the small lattice model to the larger lattice model. Specifically, we interpolate between the Hamiltonians $\H(s = 0)$ and $\H(s = 1)$ defined by
\begin{equation} \label{eqn:nucleation-growth-step}
    \H(s) = J \sum_{\text{bulk}} ZZ + s J \sum_{\text{edges}} ZZ - h  \sum_i X_i .
\end{equation}
Thus, $\H(s = 0)$ corresponds to the initial/previous lattice, and $\H(s = 1)$ corresponds to the subsequent larger lattice. In practice, we use a discretized schedule
\begin{equation} \label{eqn:schedule}
    s_k = \sin^2 \left( \frac{\pi}{2} \frac{k + 1/2}{N} \right) 
\end{equation}
where $k = 0, ..., N - 1$ with $N$ being the number of Trotter steps. Each Trotter step is performed for time
\begin{equation} \label{eqn:trotter-time-step}
    \Delta t = T / N
\end{equation}
for a total evolution time $T$. Note that the $k + 1/2$ in~\eqref{eqn:schedule} corresponds to the midpoint evaluation of the continuous schedule
\begin{equation} \label{eqn:schedule-continuous}
    s(t) = \sin^2 \pi t / 2 T 
\end{equation}
for each Trotter step. Figure~\ref{fig:nucleation-circuit} shows two nucleation steps, adding two qubits to the edge of the one-dimensional chain each time. The nucleation step~\eqref{eqn:nucleation-growth-step} continues until the final lattice Hamiltonian is reached.

Correctness of the algorithm is given by the adiabatic theorem, which asserts that we remain in the ground state of~\eqref{eqn:nucleation-growth-step} as long as $s$ is varied slowly. In the digitized implementation, this means taking the total evolution time $T$ to be sufficiently large, and the total number of Trotter steps to be sufficiently large such that we approximate the adiabatic evolution. 
First, consider the required time $T$. A classic bound for this~\cite[Theorem 3]{Jansen_Ruskai_Seiler_2007} states that
\begin{equation}  \label{eqn:bound}
    A(1) \le \frac{1}{T} \int_{0}^{1} \left[ 
        \frac{|| \ddot{\H} ||}{\Delta^2} + 7 \frac{|| \dot{\H}||^2}{\Delta^3} \, du
    \right] .
\end{equation}
Here, the quantity $A(s)$ is related to fidelity via
\begin{equation} \label{eqn:fidelity-bound}
    1 - F(s) \le A(s)^2
\end{equation}
where
\begin{equation} \label{eqn:fidelity}
    F(s) := | \langle E(s) | \psi(s) \rangle | ^2
\end{equation}
is the fidelity between the exact instantaneous ground state $| E(s) \rangle$ and the finite-time evolved state $|\psi (s) \rangle$, 
and $\Delta$ is the gap between the ground and first excited state. Note that~\eqref{eqn:bound} is exactly~\cite[Theorem 3]{Jansen_Ruskai_Seiler_2007} specialized to the case of a single nondegenerate ground state; the boundary terms of~\cite[Theorem 3]{Jansen_Ruskai_Seiler_2007} vanish via the choice of schedule~\eqref{eqn:schedule-continuous}. Also note that derivatives in~\eqref{eqn:bound} are taken with respect to the dimensionless time parameter
\begin{equation} \label{eqn:dimensionless-time-u}
    u := t / T ,
\end{equation}
and that $\Delta = \Delta(u)$.

\begin{theorem}{\cite[Theorem 3]{Jansen_Ruskai_Seiler_2007}} \label{thm:adiabatic-time}
    To achieve fidelity $F(1) \ge 1 - \epsilon$ with nucleation, with $F$ defined in~\eqref{eqn:fidelity} as the fidelity between the exact instantaneous ground state and the nucleation state, it is sufficient to take total evolution time
    \begin{equation} \label{eqn:adiabatic-time}
        T \ge \frac{1}{\sqrt{\epsilon}} \int_{0}^{1} \left[ 
            \frac{|| \ddot{\H} ||}{\Delta^2} + 7 \frac{|| \dot{\H}||^2}{\Delta^3}  
        \right] \, du .
    \end{equation}
    Here, $\Delta = \Delta(u)$ is the gap between the ground and first excited state, $u = t / T$ is a dimensionless time parameter, and derivatives of $\H$ are with respect to $u$. It is assumed that $\H$ is twice-differentiable along the adiabatic path and has a unique ground state with non-vanishing gap $\Delta$.
\end{theorem}

\begin{proof}
    Follows directly from~\eqref{eqn:bound} and~\eqref{eqn:fidelity-bound}.
\end{proof}

While it is possible to bound~\eqref{eqn:adiabatic-time} in terms of the minimum gap
\begin{equation} \label{eqn:min-gap}
    \Delta_\text{min} := \min_{0 \le u \le 1} \Delta (u)
\end{equation}
for particular Hamiltonians, Theorem~\ref{thm:adiabatic-time} and~\eqref{eqn:adiabatic-time} provide the general behavior. As one example, we can consider the transverse-field Ising model~\eqref{eqn:tfim} and obtain
\begin{equation} \label{eqn:adiabatic-time-tfim}
    T \ge \frac{1}{\sqrt{\epsilon}} \left[ \frac{2 \pi |J|}{\Delta_\text{min}^2} + \frac{7 \pi^2 J^2}{2 \Delta_\text{min}^3 }
    \right] .
\end{equation}
This follows from Theorem~\ref{thm:adiabatic-time} and~\eqref{eqn:tfim}. For completeness, a detailed derivation is provided in Appendix~\ref{sec:proof-of-tfim-adiabatic-time}.

We now move on to the Trotterization step of nucleation and determine the number of Trotter steps $N$ required. We consider Hamiltonians that can be partitioned into groups with mutually commuting terms
\begin{equation} \label{eqn:hamiltonian-AB}
    \H(u) = A(u) + B(u)
\end{equation}
such as the TFIM~\eqref{eqn:tfim}, and we consider the first-order Trotter formula
\begin{equation}
    U_N := \prod_{k = 0}^{N - 1} e^{- i A(u_k) \Delta t} e^{- i B(u_k) \Delta t} . 
\end{equation}
Here,
\begin{equation}
    u_k := \frac{k + 1/2}{N} ,
\end{equation}
and recall that $\Delta t = T / N$~\eqref{eqn:trotter-time-step}. Our goal is to bound
\begin{equation} \label{eqn:trotter-operator-error}
    E(N, T) := ||U_N - U(T)||
\end{equation}
where
\begin{equation}
    U(T) := \mathcal{T} \exp\left( -i T \int_{0}^{1} \H \, du \right)
\end{equation}
is the exact continuous-time evolution. Note that $\mathcal{T}$ is the time-ordering operator. 

\begin{theorem} \label{thm:ntrotter}
    For a Hamiltonian of the form~\eqref{eqn:hamiltonian-AB}, performing nucleation with
    \begin{equation} \label{eqn:ntrotter}
        N \ge \frac{1}{\eta} \left[ 
            \frac{T}{4} \Lambda_1 + \frac{T^2}{2} \Lambda_2
        \right]
    \end{equation}
    first order Trotter steps per nucleation stage is sufficient to achieve operator error~\eqref{eqn:trotter-operator-error} at most $\eta > 0$, where
    \begin{equation} \label{eqn:def-lambda1}
        \Lambda_1 := \max_{0 \le u \le 1} || \dot{\H} ||
    \end{equation}
    and
    \begin{equation} \label{eqn:def-lambda2}
        \Lambda_2 := \max_{0 \le u \le 1} || [ A(u), B(u) ] || . 
    \end{equation}
\end{theorem}

\begin{figure}
    \centering
    \includegraphics[width=\linewidth]{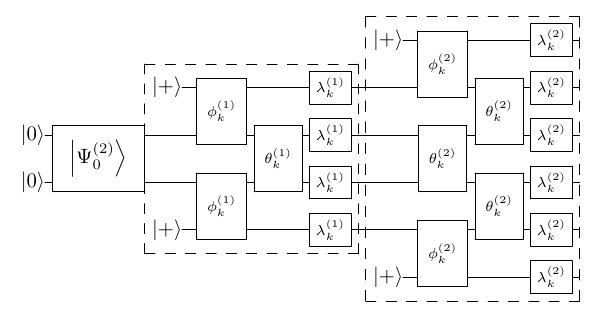}
    \caption{Sketch of the quantum circuit for optimized nucleation. The structure of the circuit is identical to the nucleation circuit in Fig.~\ref{fig:nucleation-circuit}, but here we allow variational optimization of the parameters $\boldsymbol{\theta}, \boldsymbol{\phi}$, and $\boldsymbol{\lambda}$ to minimize the energy. The subscript $k$ indicates the number of layers $k = 1, ..., L$ in each nucleation stage. That is, each dashed box repeats $L$ times, where $L$ is a hyperparameter, with independent parameters in each repetition.}
    \label{fig:optimized-nucleation-circuit}
\end{figure}

\begin{proof}
    There are two sources of error. The first comes from ``freezing'' $\H$ during each time step, and the second comes from using Trotterization within each time step. 

    Let us start with the first source of error. Over one interval centered at $t_k$, the exact propagator differs from the Trotterized operator by
    \begin{equation} \label{eqn:duhamel}
        \epsilon_k \le \int_{t_k - \Delta t / 2}^{t_k + \Delta_t / 2} || \H(t) - \H(t_k) || \, dt 
    \end{equation}
    by Duhamel's principle~\cite[Proposition 15.4]{lin_quantum_nodate}. Using~\eqref{eqn:def-lambda1}, we have
    \begin{equation}
        \left\Vert \frac{d \H}{dt} \right\Vert = \frac{1}{T} || \dot{\H} || \le \frac{\Lambda_1}{T} .
    \end{equation}
    Thus, we have that
    \begin{align}
        || H(t) - H(t_k) || &= \left\Vert \int_{t_k}^{t} \frac{d \H}{dt}  \, dt \right\Vert \\
        &\le \int_{t_k}^{t} \left\Vert  \frac{d \H}{dt}   \right\Vert \, dt \\ 
        &\le \frac{\Lambda_1}{T} |t - t_k | .
    \end{align}
    From~\eqref{eqn:duhamel}, we thus see that
    \begin{equation}
        \epsilon_k \le \frac{\Lambda_1 \Delta t^2}{4 T} .
    \end{equation}
    The total error over all $N$ intervals thus satisfies
    \begin{equation} \label{eqn:error1}
        \epsilon \le N \epsilon_k \le \frac{\Lambda_1 T}{4 N} .
    \end{equation}

    Now we consider the second source of error from Trotterizing each growth step. From~\cite{childs_theory_2021}, the standard first-order Trotter error is
    \begin{equation}
        \left\Vert e^{-i (A + B) \Delta t} - e^{-i A \Delta t} e^{- i B \Delta t} \right\Vert \le \frac{\Delta t^2}{2} || [A, B] || .
    \end{equation}
    Using~\eqref{eqn:def-lambda2}, one step contributes a maximum error of $\Delta t^2 \Lambda_2 / 2$, and over $N$ steps we have a total error of
    \begin{equation} \label{eqn:error2}
        \delta \le \frac{T^2}{2 N} \Lambda_2 .
    \end{equation}

    Combining~\eqref{eqn:error1} and~\eqref{eqn:error2}, we thus have
    \begin{equation}
        E(N, T) \le \frac{1}{N} \left[ 
        \frac{T}{4} \Lambda_1 + \frac{T^2}{2} \Lambda_2
        \right] ,
    \end{equation}
    from which~\eqref{eqn:ntrotter} follows.
\end{proof}

Again it is possible to obtain more detailed bounds for particular Hamiltonians. For example, for the TFIM~\eqref{eqn:tfim} one can show that
\begin{equation} \label{eqn:ntrotter-tfim}
    N \ge \frac{1}{\eta} \left[
        \frac{\pi |J| T}{4} + 2 |Jh| (n - 1) T^2
    \right]
\end{equation}
where $n$ is the number of qubits (lattice sites). For completeness, we prove~\eqref{eqn:ntrotter-tfim} in Appendix~\ref{sec:ntrotter-tfim}.

We remark that bounds for both adiabatic and Trotter evolution can significantly overestimate requirements for convergence. Indeed, recent work~\cite{kovalsky_self-healing_2023} has shown that first-order Trotterization of a complete adiabatic evolution exhibits a ``self-healing'' mechanism and has a cumulative infidelity that scales as $O(T^{-2} \Delta t^2)$ instead of the $O(T^2 \Delta t^2)$ expected from general Trotter error bounds. To probe this, we perform direct numerical simulations of nucleation in Sec.~\ref{sec:numerical-results}, and find that the actual $T$ and $N$ required for convergence markedly smaller than what is expected from the bounds in Theorem~\ref{thm:adiabatic-time} and Theorem~\ref{thm:ntrotter}. Nonetheless, these bounds provide rigorous theoretical guarantees for the correctness and convergence of nucleation.  Although nucleation requires multiple growth stages, each stage changes the Hamiltonian only locally by coupling newly added sites to the existing lattice. Consequently, each growth stage addresses an incremental local state-preparation problem, rather than repeating the preparation of the entire many-body state.  In the numerical examples studied in Sec.~\ref{sec:numerical-results}, we find that the practical resources required are far smaller than the general worst-case bounds. Establishing the asymptotic scaling of the total cost with system size is left for future work. Optimized nucleation, introduced next, provides a practical route to reducing the gate overhead associated with repeated adiabatic growth.

\section{Optimized nucleation} \label{sec:optimized-nucleation}

\begin{figure}
    \centering
    \includegraphics[width=\linewidth]{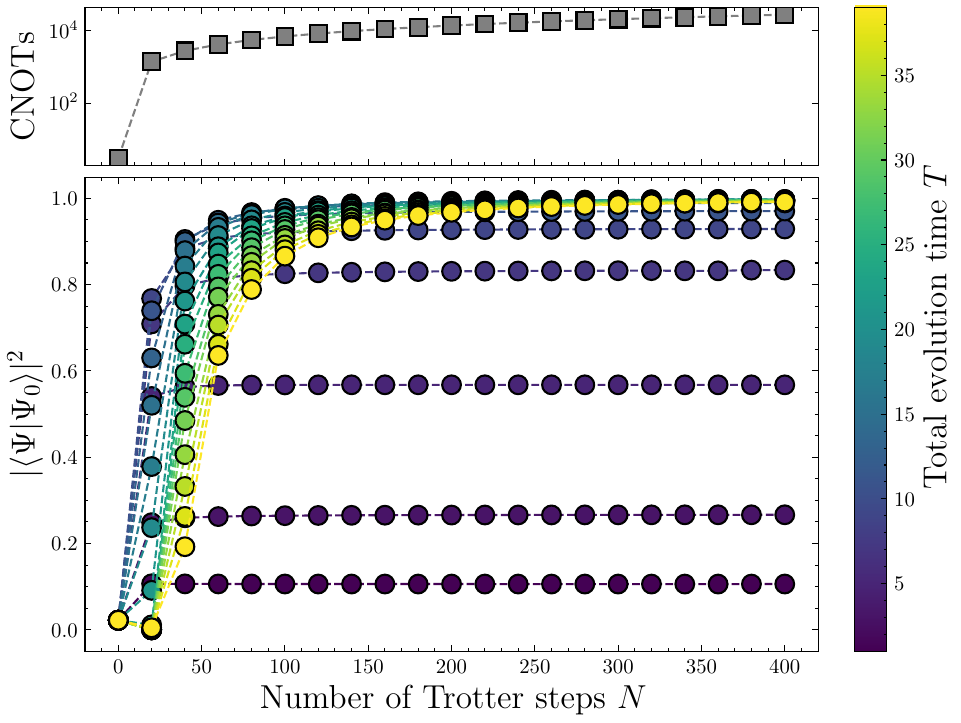}
    \caption{Illustration of nucleation from a two-qubit to a twelve-qubit one-dimensional TFIM~\eqref{eqn:tfim} with $J = 1$ and $h = 1/2$. The colorbar shows the total evolution time $T$ (in units of $1 / J$), and the horizontal axis shows the number of Trotter steps $N$. Initial state overlap is shown at $N = 0$ Trotter steps. As can be seen, the nucleation state $|\Psi\rangle$ approaches the true ground state $|\Psi_0\rangle$ as the total evolution time and number of Trotter steps increases. The highest overlap we achieve is $0.997$, starting from the initial overlap (for all $T$, $N$) of 0.022.}
    \label{fig:1dtfim}
\end{figure}

For near-term and even early fault-tolerant quantum computers, Trotterization can be prohibitively expensive and lead to performance degradation as noise accumulates. Motivated by this, we introduce optimized nucleation, shown in Fig.~\ref{fig:optimized-nucleation-circuit}, to reduce the circuit depth in nucleation, at the cost of increased runtime for variational optimization. Optimized nucleation employs the same circuit structure as nucleation (Fig.~\ref{fig:nucleation-circuit}), but we introduce variational parameters $\boldsymbol{\theta}, \boldsymbol{\phi}$, and $\boldsymbol{\lambda}$ which are optimized to minimize energy. Angles $\boldsymbol{\phi}$ capture the adiabatic schedule --- in other words, the schedule itself is found by energy minimization rather than being explicitly defined as in~\eqref{eqn:schedule}. As such,  $\boldsymbol{\phi}$ parameterizes the new edge gates that are introduced at each nucleation stage. Angles $\boldsymbol{\theta}$ and $\boldsymbol{\lambda}$ parameterize the bulk couplings and local rotations, respectively. Optimized nucleation takes a hyperparameter $L$, the number of layers in each nucleation stage, which sets the number of variational parameters. In Fig.~\ref{fig:optimized-nucleation-circuit}, this is indicated schematically with the subscript $k$, where $k = 1, ..., L$. Parameters in each nucleation stage are independent, and we use a superscript indexing the nucleation stage in Fig.~\ref{fig:nucleation-circuit} to indicate this. We remark again that optimized nucleation shares similarities with the variational procedure of~\cite{olund_adiabatic_2020} which uses tensor network methods for optimization (specifically, the SVD rule for computing the locally optimal unitary to add to the circuit). Our method differs in the parameter initialization by using ``warm start'' values from Trotterization, optimization procedure by using layerwise optimization, and again in the nucleation growth stages where we arithmetically grow the lattice size instead of geometric growth.

\begin{figure}
    \centering
    \includegraphics[width=\linewidth]{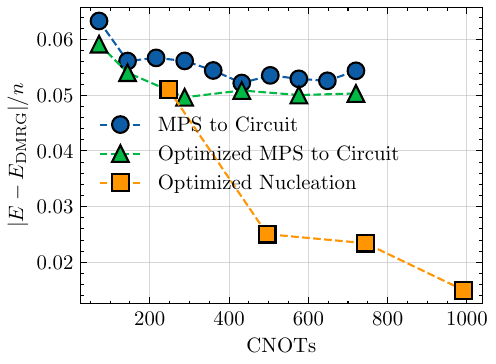}
    \caption{Error vs number of two-qubit gates for the MPS to circuit method and optimized nucleation. These results are for a $5 \times 5$ transverse-field Ising model with $J = 1.0$ and $h = 2.5$. Errors are computed relative to DMRG with bond dimension $\chi_{\text{DMRG}} = 196$. Numerical simulations for computing energies are done with an MPS with bond dimension $\chi = 100$ for both the MPS to circuit and optimized nucleation methods. As shown, optimized nucleation provides a lower energy error than the MPS to circuit method, even when we allow for parameter optimization after the MPS to circuit method. For fairness of comparison, both optimized nucleation and optimized MPS to circuit are provided with the same number of optimization iterations.}
    \label{fig:2dtfim-compare-to-mps}
\end{figure}

Optimized nucleation provides a tradeoff between full nucleation, which has no optimization overhead but potentially high gate overhead, and full variational optimization, which has low gate overhead but potentially high optimization overhead. The nucleation circuit provides a structure for optimization that is guaranteed to converge to the ground state for some set of parameters and large enough $L$. This reduces the complexity of variationally finding the circuit structure, which is a bottleneck for many variational algorithms at scale~\cite{cerezo_variational_2021}. Additionally, nucleation provides an initial parameterization --- i.e., taking angles as in Fig.~\ref{fig:nucleation-circuit} for $N = L$ layers. Good initial parameterization has been found to be important for overcoming barren plateaus~\cite{mcclean_barren_2018}, and this is often nontrivial in other approaches to variational optimization such as the brickwork ansatz~\cite{bravo-prieto_variational_2023}.

Many methods have been developed for optimizing angles in parameterized quantum circuits. Here, we employ a layer-wise optimization approach. In this approach, we start from $L = 1$ layer and optimize the parameters. Once the energy is minimized for $L = 1$, we move to $L = 2$. The initial set of parameters is taken to be the optimal parameters for $L = 1$ and the Trotterized initial parameters for $L = 2$. From these initial values, all angles are then optimized to again minimize the energy. This process continues until the final number of layers $L$ is reached, or until the energy reaches a pre-determined threshold. We remark that this layer-wise optimization procedure is similar to that of Ref.~\cite{skolik_layerwise_2021} and Ref.~\cite{larose_mixer-phaser_2022}. In practice we find the number of layers $L$ required for variational optimization can be significantly less than what is expected from Trotterization, leading to significant reductions in gate counts. These results are shown in Sec.~\ref{sec:numerical-results} and Fig.~\ref{fig:2dtfim-compare-to-mps}, in which optimized nucleation achieves a lower energy error per CNOT than the MPS to circuit method, a current leading approach.

\section{Numerical results} \label{sec:numerical-results}

We first demonstrate the correctness of nucleation and characterize the number of Trotter steps required for convergence for a particular example. Specifically, we consider a one-dimensional transverse-field Ising model~\eqref{eqn:tfim} with $J = 1$ and $h = 1/2$, and perform nucleation for various evolution times $T$ and (first-order) Trotter steps $N$. The results, shown in Fig.~\ref{fig:1dtfim}, show that nucleation converges to the true ground state. Interestingly, this convergence occurs with both $T$ and $N$ significantly smaller than what is expected from the bounds in Theorem~\ref{thm:adiabatic-time} and Theorem~\ref{thm:ntrotter}. Indeed, based on the minimum gap (see Appendix~\ref{sec:additional-numerical-results} and Fig.~\ref{fig:gap}), we would expect $T \ge 2.57 \cdot 10^{3}$ and $N \ge 2.30 \cdot 10^{10}$. As shown, in practice we achieve high overlap with the true ground state with $T \approx 40$ and $N \approx 400$ Trotter steps per stage, or $N \approx 2000$ total Trotter steps over all five growth stages. The highest overlap we achieve is $0.997$, starting from the initial overlap (for all $T$, $N$) of 0.022.

\begin{figure}
    \centering
    \includegraphics[width=\linewidth]{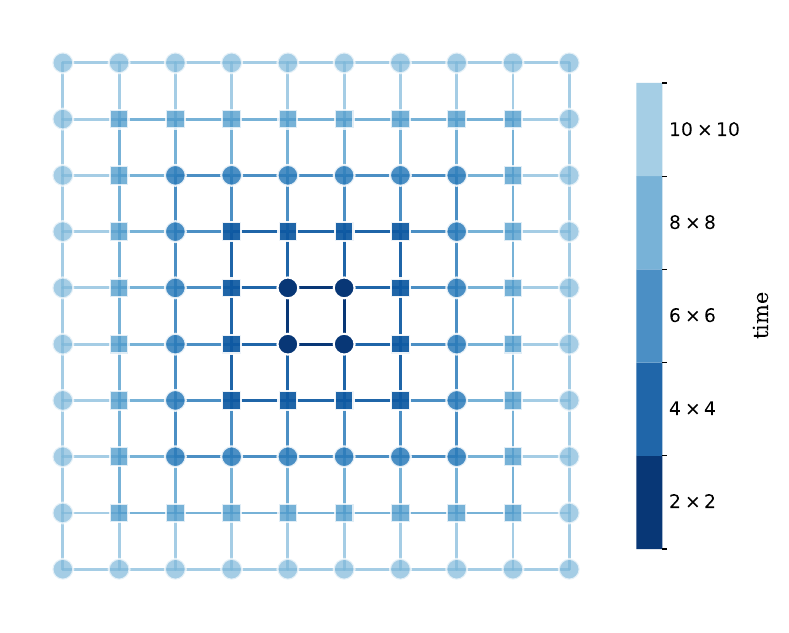}
    \caption{Illustration of two-dimensional nucleation growth, starting with a core $2 \times 2$ lattice and expanding in time to a $4 \times 4$, $6 \times 6$, ..., and a $10 \times 10$ lattice. Different markers are used between growth layers for clarity, and colors progress from darker to lighter as time increases. While this growth pattern is perhaps the most natural to consider, several alternatives are possible, e.g. the column-wise growth shown in Appendix~\ref{sec:additional-numerical-results} Fig.~\ref{fig:nucleation-growth-columns}.}
    \label{fig:nucleation-growth-2d}
\end{figure}

We next consider optimized nucleation. Here, we use a two-dimensional TFIM~\eqref{eqn:tfim} with $J = 1.0$, $h = 2.5$ on a $5 \times 5$ lattice. We compare optimized nucleation to decomposing a matrix product state (MPS) to a circuit (``MPS to circuit''), a current state-of-the-art method for quantum state preparation~\cite{ran_encoding_2020,mpstocircuit2025}. The results, shown in Fig.~\ref{fig:2dtfim-compare-to-mps}, demonstrate that optimized nucleation can reach a lower energy error, measured by the difference with DMRG energy, with fewer two-qubit gates than MPS to circuit. This even holds when we allow for parameter optimization after MPS to circuit (``Optimized MPS to circuit'' in Fig.~\ref{fig:2dtfim-compare-to-mps}). Here, we use the same L-BFGS-B optimizer~\cite{Liu_Nocedal_1989} implemented in SciPy~\cite{virtanen_scipy_2020} for both optimizations, with the same number of optimization iterations for each. We remark that the MPS to circuit method excels for one-dimensional examples and many others, and the two-dimensional example here was chosen specifically as a case where optimized nucleation can provide improved performance. While a full comparison will depend on more examples and particular problem details, we expect that (optimized) nucleation will provide superior performance for certain two- and higher-dimensional lattice models as we see in this instance.

In two dimensions, unlike the one-dimensional case of Fig.~\ref{fig:nucleation-circuit}, there are several more choices for how to grow the lattice in each nucleation growth stage. A natural choice is illustrated in Fig.~\ref{fig:nucleation-growth-2d}, where a small core nucleus grows symmetrically outward from a $2 \times 2$ lattice to a $4 \times 4$ lattice, then to $6 \times 6$ lattice, and so on. Another strategy is to start with a one-dimensional model (e.g., a column) that grows symmetrically outward (left and right) until the full lattice size is reached. This strategy is illustrated in Appendix~\ref{sec:additional-numerical-results} and Fig.~\ref{fig:nucleation-growth-columns}. For ease of implementation, we use this symmetric column growth strategy in the optimized nucleation experiment of Fig.~\ref{fig:2dtfim-compare-to-mps}, though either could be used. The most appropriate choice of nucleation growth strategy will likely depend on details of both the problem and quantum computer, and is left to future work.

% \section{Quantum hardware results}

\section{Conclusion}

We introduced an algorithm for digital quantum state preparation of lattice Hamiltonians based on nucleation. Our algorithm works by preparing an initial state of a small lattice Hamiltonian, then performing growth stages which adiabatically couple the small model to progressively larger lattices. To implement this in a quantum circuit, we Trotterize this adiabatic evolution, providing an end-to-end algorithm for initial state preparation on quantum computers. In addition to proving bounds on the total evolution time and number of Trotter steps required, we also implemented the algorithm numerically. Our numerical results show that nucleation can prepare good initial states with orders of magnitude lower resources than expected from these worst-case bounds. Additionally, we introduced optimized nucleation for quantum state preparation on current/near-term quantum computers, and showed in numerical experiments that this method can outperform a leading state preparation method for a two-dimensional lattice Hamiltonian.

Our paper creates several ideas and directions for future work with current/near-term quantum computers~\cite{preskill_quantum_2018}, Megaquop computers~\cite{preskill_beyond_2025}, and future fault-tolerant quantum computers~\cite{eisert_mind_2025}. In the near-term setting, one can consider implementing optimized nucleation on quantum hardware, and performing a detailed benchmarking study with respect to other state preparation methods. As we have seen, we expect that nucleation can achieve good state preparation with fewer two-qubit gates for certain two- and higher-dimensional lattice models. As two-qubit gates are the primary source of error on current (and likely Megaquop) quantum computers, it could be possible to prepare higher-fidelity initial states. With the use of error mitigation~\cite{larose_mitiq_2022,takagi_fundamental_2022}, it could even be possible to use these initial states for Trotterized time evolution for studying dynamics, or as initial states in other quantum algorithms. In the fault-tolerant setting, future work can include analyzing the error correction resources (physical qubits, $T$ gates, etc.) for nucleation, and again benchmarking against other methods. Notably, nucleation can provide rigorous worst-case upper bounds on the gate counts for the cost of initial state preparation in fault-tolerant quantum algorithms, which is often omitted. Ultimately, as demonstrated and discussed, we believe that nucleation will prove to be a valuable method for preparing initial states of lattice Hamiltonians on digital quantum computers, enabling studies of high-energy physics~\cite{bauer_quantum_2023}, field theories~\cite{jordan_quantum_2011,jordan_quantum_2012,jordan_quantum_2014,nuqs_collaboration_general_2019}, and similar areas with quantum computers.

\vspace{1em}

\section*{Software and data availability}

Software implementing (optimized) nucleation and reproducing all numerical results is available at~\cite{nucleation_github}.

\section*{Acknowledgments}

We acknowledge support from the Department of Energy under grant number DE-SC0023658. RL acknowledges the Aspen Center for Physics, which is supported by National Science Foundation grant PHY-2210452, at which a portion of this work was completed. DL has also received support through Department of Energy grants DE-SC0013365, DE-SC0023175, and DE-SC0026198.  We acknowledge the use of generative AI (Anthropic Claude Opus 5 and OpenAI GPT 5.6 Sol) for help with numerical implementations of nucleation and optimized nucleation, preparing Fig.~\ref{fig:nucleation-growth-2d} and Fig.~\ref{fig:nucleation-growth-columns}, and as a general tool for suggesting references, proof techniques, and reviewing the manuscript and code for correctness.

\bibliographystyle{apsrev4-1}
\bibliography{refs}

\newpage

\appendix

\section{Proof of Eqn.~\eqref{eqn:adiabatic-time-tfim}} \label{sec:proof-of-tfim-adiabatic-time}

Here we prove~\eqref{eqn:adiabatic-time-tfim} from Theorem~\ref{thm:adiabatic-time} and the transverse-field Ising model (TFIM)~\eqref{eqn:tfim}. From Theorem~\ref{thm:adiabatic-time} (Eqn.~\eqref{eqn:adiabatic-time}), we need to evaluate $\dot{\H}$ and $\ddot{\H}$, where dots indicate derivatives with respect to the dimensionless time parameter $u \equiv t / T$~\eqref{eqn:dimensionless-time-u}. We see from~\eqref{eqn:nucleation-growth-step} that 
\begin{equation} \label{eqn:hdot-t}
    ||\dot{\H}|| = \left\Vert \frac{d\H}{du} \right\Vert = \left\Vert J \sum_{\text{edges}} ZZ \frac{ds}{du} \right\Vert .
\end{equation}
The schedule~\eqref{eqn:schedule-continuous} expressed in terms of $u = t / T$ is
\begin{equation}
    s(u) = \sin^2 ( \pi u / 2 ) ,
\end{equation}
from which we have
\begin{equation}
    \frac{ds}{du} = 2 \sin(\pi u / 2) \cos(\pi u / 2) \frac{\pi}{2} = \frac{\pi}{2} \sin ( \pi u ) ,
\end{equation}
where the second equality uses the double angle identity. Thus~\eqref{eqn:hdot-t} is
\begin{equation} \label{eqn:for-lambda1}
    ||\dot{\H}|| = \left\Vert \sum_{\text{edges}} ZZ \right\Vert \frac{\pi}{2} | J \sin ( \pi u) |  = \pi | J \sin ( \pi u) | . 
\end{equation}
Note that $\left\Vert \sum_{\text{edges}} ZZ \right\Vert = 2$ because there are two edges added in each nucleation growth stage in one dimension, as illustrated in Fig.~\ref{fig:nucleation-circuit}. 
This implies that
\begin{equation} \label{eqn:hdot}
    \int_{0}^{1}  7 \frac{|| \dot{\H}||^2}{\Delta(u)^3} \, du \le \frac{7 \pi^2 J^2}{\Delta_\text{min}^3} \int_{0}^{1}  \sin^2 (\pi u) \, du = \frac{7 \pi^2 J^2}{2 \Delta_\text{min}^3}  .
\end{equation}
Further, we have
\begin{equation}
    \frac{d^2s}{du^2} = \frac{\pi^2}{2} \cos ( \pi u ) ,
\end{equation}
and thus
\begin{align*}
    ||\ddot{\H}|| = \pi ^ 2|J\cos(\pi u)|
\end{align*}
by similar reasoning as above. This implies that
\begin{equation} \label{eqn:hddot}
    \int_{0}^{1} \frac{||\ddot{\H}||}{\Delta(u)^2} \, du \le \frac{\pi^2 |J| }{\Delta_\text{min}^2} \int_{0}^{1} | \cos ( \pi u ) | \, du = \frac{2 \pi |J|}{\Delta_\text{min}^2} .
\end{equation}
Substituting~\eqref{eqn:hdot} and~\eqref{eqn:hddot} into Theorem~\ref{thm:adiabatic-time} recovers~\eqref{eqn:adiabatic-time-tfim}.

We remark that we numerically compute the gap $\Delta(u)$ for particular Hamiltonian instances in Appendix~\ref{sec:additional-numerical-results} and compare these bounds to the required resources in practice.

\begin{figure}
    \centering
    \includegraphics[width=\linewidth]{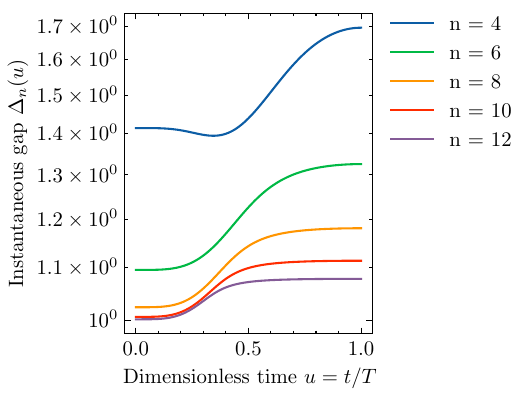}
    \caption{The gap $\Delta = \Delta(u)$ appearing in Theorem~\ref{thm:adiabatic-time} for the one-dimensional transverse-field Ising model~\eqref{eqn:tfim} with $J = 1$ and $h = 1/2$. This corresponds to the setting of Fig.~\ref{fig:1dtfim}. As shown, the minimum gap~\eqref{eqn:min-gap} is $\Delta_\text{min} \simeq 1.001$. Taking $\epsilon = 10^{-3}$, which is comparable to the accuracy achieved in Fig.~\ref{fig:1dtfim}, leads to the bound from~\eqref{eqn:adiabatic-time-tfim} of $T \ge 2.57 \cdot 10^{3}$ and the bound from~\eqref{eqn:ntrotter-tfim} of $N \ge 2.30 \cdot 10^{10}$. Compared to Fig.~\ref{fig:1dtfim} where we see that nucleation provides good initial state preparation for $T \simeq 40$ and $N \simeq 2000$ total Trotter steps, we see that the analytical bounds in Theorem~\ref{thm:adiabatic-time} and Theorem~\ref{thm:ntrotter} are sufficient but can provide orders of magnitude larger estimates than what is necessary for convergence.}
    \label{fig:gap}
\end{figure}

\section{Proof of Eqn.~\eqref{eqn:ntrotter-tfim}} \label{sec:ntrotter-tfim}

Here we prove~\eqref{eqn:ntrotter} from Theorem~\ref{thm:ntrotter} and the transverse-field Ising model~\eqref{eqn:tfim}. We first evaluate~\eqref{eqn:def-lambda1}. Using~\eqref{eqn:for-lambda1}, we see that
\begin{equation} \label{eqn:lambda1-tfim}
    \Lambda_1 := \max_{0 \le u \le 1} || \dot{\H} || = \pi |J| .
\end{equation}

Next, we evaluate~\eqref{eqn:def-lambda2}. To do so, for the TFIM~\eqref{eqn:nucleation-growth-step} we have
\begin{equation}
    A(u) = \sum_{i = 1}^{n - 1} c_i(u) Z_i Z_{i + 1}
\end{equation}
where
\begin{equation}
    c_i(u) := \begin{cases}
        J           & i \in \text{bulk} \\
        s(u) J & i \in \text{edges}
    \end{cases}
\end{equation}
and
\begin{equation}
    B(u) = -h\sum_{i = 1}^{n} X_i .
\end{equation}
Using $[Z, X] = 2 ZX$, we have
\begin{equation}
    [A(u), B(u)] = 2 h \sum_{i = 1}^{n - 1} c_i (u) (X_i + X_{i + 1}) Z_i Z_{i + 1} .
\end{equation}
Thus we can evaluate~\eqref{eqn:def-lambda2} as
\begin{equation} \label{eqn:lambda2-tfim}
    \Lambda_2 := \max_{0 \le u \le 1} || [ A(u), B(u) ] || \le 4 |h J| (n - 1). 
\end{equation}
Substituting~\eqref{eqn:lambda1-tfim} and~\eqref{eqn:lambda2-tfim} into Theorem~\ref{thm:ntrotter} recovers~\eqref{eqn:ntrotter-tfim}.

\section{Additional numerical results\\and details} \label{sec:additional-numerical-results}

\begin{figure}
    \centering
    \includegraphics[width=\linewidth]{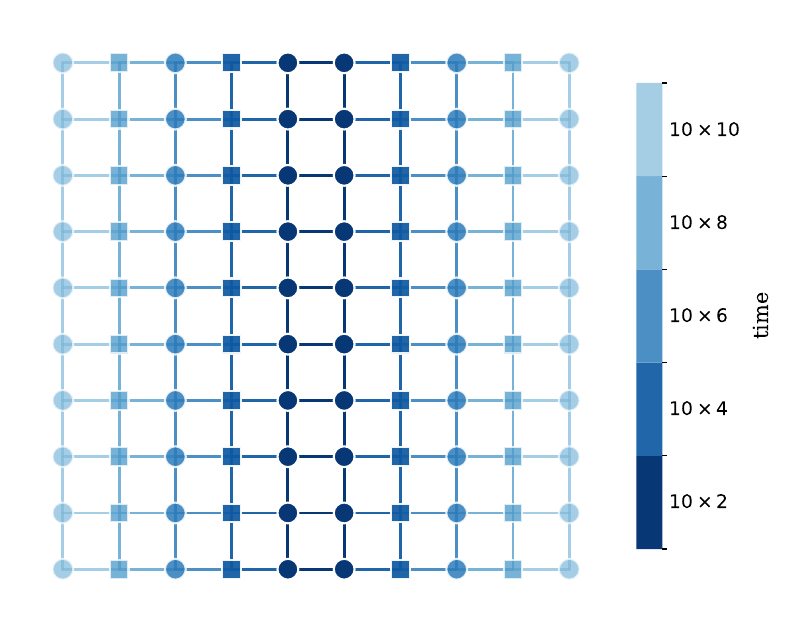}
    \caption{Alternative two-dimensional growth strategy to Fig.~\ref{fig:nucleation-growth-2d}, expanding in columns symmetrically on the left and right. For ease of implementation, this growth strategy was used in Fig.~\ref{fig:2dtfim-compare-to-mps}.}
    \label{fig:nucleation-growth-columns}
\end{figure}

Here we numerically compute the minimum gap~\eqref{eqn:min-gap} of the TFIM~\eqref{eqn:tfim} to evaluate the bound~\eqref{eqn:adiabatic-time-tfim} for the parameter setting of Fig.~\ref{fig:1dtfim} ($J = 1$, $h = 1/2$). To do so, we define a discrete grid of points $u_k$ and evaluate $\Delta(u_k)$, taking the minimum. In practice, we use $101$ grid points in the interval $u \in [0, 1]$ and diagonalize the Hamiltonian projected into the even parity subspace. To diagonalize, we use the Implicitly Restarted Lanczos Method~\cite{Lehoucq_Sorensen_Yang_1998} implemented in SciPy~\cite{virtanen_scipy_2020}. After finding the minimum $\Delta(u_k)$ along the grid, we further minimize $\Delta$ around this $u_k$, then take the smallest found gap during minimization as $\Delta_\text{min}$. This process is performed for all $n$ encountered over the full nucleation algorithm (i.e., over all nucleation growth stages). The results are shown in Fig.~\ref{fig:gap}.

In Fig.~\ref{fig:gap}, we see that the minimum gap~\eqref{eqn:min-gap} for nucleation from the $n = 2$ qubit to $n = 12$ qubit one-dimensional TFIM corresponding to Fig.~\ref{fig:1dtfim} is $\Delta_\text{min} \simeq 1.001$. Taking $\epsilon = 10^{-3}$ and this $\Delta_\text{min}$ substituted into the Theorem~\ref{thm:adiabatic-time} bound specialized to the TFIM (i.e., Eqn.~\eqref{eqn:adiabatic-time-tfim}), we obtain $T \ge 2.57 \cdot 10^{3}$. Recall that, as shown in Fig.~\ref{fig:1dtfim}, we saw that $T \simeq 40$ with $N \simeq 2000$ total Trotter steps was sufficient to produce high overlap with the same problem instance. (Note that Fig.~\ref{fig:1dtfim} performs $400$ Trotter steps per stage and there are five total stages.) Further, we also numerically evaluate the bound on the number of Trotter steps required in Theorem~\ref{thm:ntrotter} specialized to the problem instance (i.e., Eqn.~\eqref{eqn:ntrotter-tfim}). Doing so with $\eta = \sqrt{\epsilon} / 2$ yields $N \ge 2.30 \cdot 10^{10}$ total Trotter steps. 
This illustrates that, while the bounds of Theorem~\ref{thm:adiabatic-time} and Theorem~\ref{thm:ntrotter} are sufficient, they are not strictly necessary, and nucleation may produce good initial states with (significantly) fewer resources than expected from these bounds.

Finally, in Fig.~\ref{fig:nucleation-growth-columns} we shown another growth strategy for two-dimensional lattice models. This provides an alternative to the strategy shown in Fig.~\ref{fig:nucleation-growth-2d}. Whereas in Fig.~\ref{fig:nucleation-circuit} the state grows radially outward from a small core in the center, in Fig.~\ref{fig:nucleation-growth-columns} the state grows symmetrically outward in columns. In practice, we used the strategy shown in Fig.~\ref{fig:nucleation-growth-columns} to produce the numerical results for optimized nucleation shown in Fig.~\ref{fig:2dtfim-compare-to-mps}, though either method, or additional methods, could be used. The best strategy will likely depend on the particular problem under consideration, and determining the optimal strategy for a given problem is a subject of future work.

\end{document}